\documentclass[preprint,
 amsmath,amssymb,
 aps, physrev,
showkeys
]{revtex4-2}
\newtheorem{theorem}{Theorem}
\newtheorem{lemma}{Lemma}
\newtheorem{proof}{Proof}

\newtheorem{corollary}{Corollary}
\usepackage{graphicx}
\usepackage{dcolumn}
\usepackage{bm}

\begin{document}

\preprint{APS/123-QED}

\title{\textbf{Quantifying system-environment synergistic information by effective information decomposition} 
}%

\author{Mingzhe Yang}
\author{Linli Pan}%
\affiliation{%
 School of Systems Science, Beijing Normal University.
}%

\author{Jiang Zhang}
 
\affiliation{
 School of Systems Science, Beijing Normal University
}%
\affiliation{
 Swarma Research
}%

\date{\today}

\begin{abstract}
What is the most crucial characteristic of a system with life activity? Currently, many theories have attempted to explain the most essential difference between living systems and general systems, such as the self-organization theory and the free energy principle, but there is a lack of a reasonable indicator that can measure to what extent a system can be regarded as a system with life characteristics, especially the lack of attention to the dynamic characteristics of life systems. In this article, we propose a new indicator at the level of dynamic mechanisms to measure the ability of a system to flexibly respond to the environment. We proved that this indicator satisfies the axiom system of multivariate information decomposition in the partial information decomposition (PID) framework. Through further disassembly and analysis of this indicator, we found that it is determined by the degree of entanglement between system and environmental variables in the dynamics and the magnitude of noise. We conducted measurements on cellular automata (CA), random Boolean networks, and real gene regulatory networks (GRN), verified its relationship with the type of CA and the Langton parameter, and identified that the feedback loops have high abilities to flexibly respond to the environment on the GRN. We also combined machine learning technology to prove that this framework can be applied in the case of unknown dynamics. 
\end{abstract}

\maketitle


\section{\label{sec:intro}Introduction}
Many complex systems exhibit life-like characteristics. Flexible robots interact seamlessly with humans, and orderly online communities give rise to innovative crowdsourced products. So what is the key difference between them and ordinary systems? The self-organization theory is the first to answer this question \cite{kauffman1993origins}. People believe that a system must have the ability of self-organization to be considered a system with life-like characteristics \cite{langton1990computation, lansing2003complex}. However, life is not just an open system with self-organization \cite{friston2007free}. More importantly, they can still maintain the stability of their own structure and function when faced with a diverse and changing environment. For example, snowflake is a self-organizing system \cite{lansing2003complex}, but it is not considered to have the activity possessed by life. Once the environmental temperature rises, the snowflake melts. If this snowflake could autonomously avoid high-temperature environments and phase transitions that would damage its own structure, then it would be considered as an adaptive system with life activity \cite{friston2007free}. This kind of adaptive systems has been discovered in a large number of papers in various fields and is a major characteristic of complex systems \cite{holland1991artificial, holland1992adaptation, holland1992complex}. 

In addition to qualitative discussions, we urgently need a formal framework to quantify and identify this unique property of life. There are already metrics to measure the degree of a system's self-organization \cite{rosas2018information,gershenson2012complexity}. However, to date, few indicators exist that assist us in measuring to what extent a system can cope with environmental changes. The proposal of the information theory of individuality \cite{krakauer2020information} is to measure the individuality of living systems from the perspective of information dynamics. It describes the individual survival of a system as maximizing the transmission of its own information over time, so mutual information and conditional mutual information are used to define the individuality of an organism. 

However, in individual information theory, the calculation of the metric (mutual information) depends on the state distribution of the observed data. The state distribution we observe depends on the initial conditions of the system and the environment, as well as the duration of the dynamical process (if the system is in a non-steady state). However, the characteristic of life's flexible response to the environment is not a property that varies with time and state, but rather reflects the characteristics of the interaction mechanisms between the system and the environment, representing a dynamical property. Therefore, when discussing the features of complex systems, we should focus on the quantities defined on their causal mechanisms, rather than the states \cite{hoel2013quantifying}. These causal mechanisms should be invariant in time. This constancy is crucial for describing the properties of the system.

Tononi and Hoel et al. \cite{oizumi2014phenomenology, hoel2013quantifying} have proposed information indicators measured at the causal mechanism level, such as effective information (EI) \cite{oizumi2014phenomenology, hoel2013quantifying}. The causal mechanism typically remains invariant with respect to state and time. It is commonly assumed that dynamics are described by Markovian transition probability matrices (TPM) \cite{hoel2013quantifying, hoel2017map}, and EI is designed as a function of TPM \cite{yuan2024emergence, zhang2024dynamical}. EI is used to measure the strength of causal effects in dynamics \cite{hoel2013quantifying, zhang2024dynamical}, and the difference in EI between macro and micro dynamics is employed to quantify the degree of emergence in complex systems \cite{hoel2013quantifying,hoel2017map}. When dynamical mechanisms are unknown, machine learning techniques can identify causal mechanisms from the data \cite{yang2024finding}. Although EI itself does not account for the interactions between the system and the environment \cite{varley2022flickering}, this does not prevent us from extending EI to develop a causal metric that incorporates environmental influences. Consequently, we can leverage this kind of indicator, along with other relevant metrics \cite{pearl2018book}, to describe causal properties and craft an indicator capable of measuring the system's responsiveness to the environment at the causal mechanism level.

Moreover, through partial information decomposition(PID) \cite{williams2010nonnegative}, we see that individuality involves redundant information and synergistic information in which the system and the environment are coupled together. Individual information theory uses the PID theory to explain the physical meaning of individuality indicators, but failed to calculate information atoms that describe the system's flexible response to the environment \cite{krakauer2020information}. Numerous methods have been proposed to calculate information atoms, yet their computational outcomes variously contravene common consensus on certain properties \cite{lizier2018information}. In this paper, we introduce EI, which requires the input variables from the previous moment to be uniformly distributed \cite{oizumi2014phenomenology, hoel2013quantifying} and starting from the existing PID axiomatic system \cite{bertschinger2013shared}, we derive a computable definition of information atoms in a three-variable system. This allows us to obtain information-theoretic metrics that precisely align with the meanings of uniqueness and synergy.

In this letter, we will define an indicator at the causal mechanism level to characterize the ability of a system to maintain its own structure and function when dealing with environmental changes. Through mathematical proofs, we will show that this indicator is precisely the synergistic information of the system itself when the system and the environment are coupled together. In our numerical experiments, we demonstrate the correlation between this indicator and the edge of chaos within cellular automata. Additionally, we apply this indicator to gene regulatory networks (GRNs), uncovering the significance of feedback loops (FBLs) and the responses of the steady states of a highly synergistic system to environmental alterations. Notably, FBLs are instrumental in carrying out the biological functions of biological systems in reacting to environmental changes \cite{alon2019introduction}. Additionally, we conducted a machine learning experiment to demonstrate that causal mechanisms can be identified using machine learning when data alone is available.

\section{\label{sec:2}Formulation}
Next, we provide the formal expressions for the system and the environment. There are three variable combinations that have a causal influence on the system's state at the next moment: the system itself, the environment, and the joint variables of the system and environment. Based on this, we define three causal mechanisms: the Individual Mechanism, the External Driving Mechanism, and the Distinct Mechanism.

\subsection{Distinct Mechanism}
We define the \textit{System} \( X \) as a subset of the \textit{World} \( U \), with respective state sets \( \Omega_X \) and \( \Omega_U \). Assuming that the dynamics of \( U \) satisfy Markovianity and are discrete, they can be described by the conditional probability \( P(U^{t+1} \mid U^t) \), where \( t \) denotes the time in this stochastic process. To exclusively measure the dynamical properties, we eliminate the influence of the World data distribution by introducing the do-operator \cite{pearl2018book}, denoted as $do(U^t \sim \mathcal{U}(\Omega_{U}))$, where $\mathcal{U}$ represents the uniform distribution.

For the stochastic process at time \( t+1 \), our analysis focuses exclusively on the temporal evolution of the target system \( X^{t+1} \). This necessitates marginalizing over extraneous variables in the global system \( U \), thereby restricting attention to the marginalized conditional probability: $P(X^{t+1}\mid U^t)$. Formally, this probability measure is obtained through state space projection:  $\forall x^{t+1} \in \Omega_X, P(X^{t+1} = x^{t+1} \mid U^t) = \sum_{\substack{u^{t+1} \in \Omega_U \\ \pi_X(u^{t+1}) = x^{t+1}}} P(U^{t+1} = u^{t+1} \mid U^t )$. Here \( \pi_X: \Omega_U \to \Omega_X \) represents the canonical projection mapping that extracts the \( X \)-component from the global state.  



\begin{figure}
    \centering
    \includegraphics[width=1.0\linewidth]{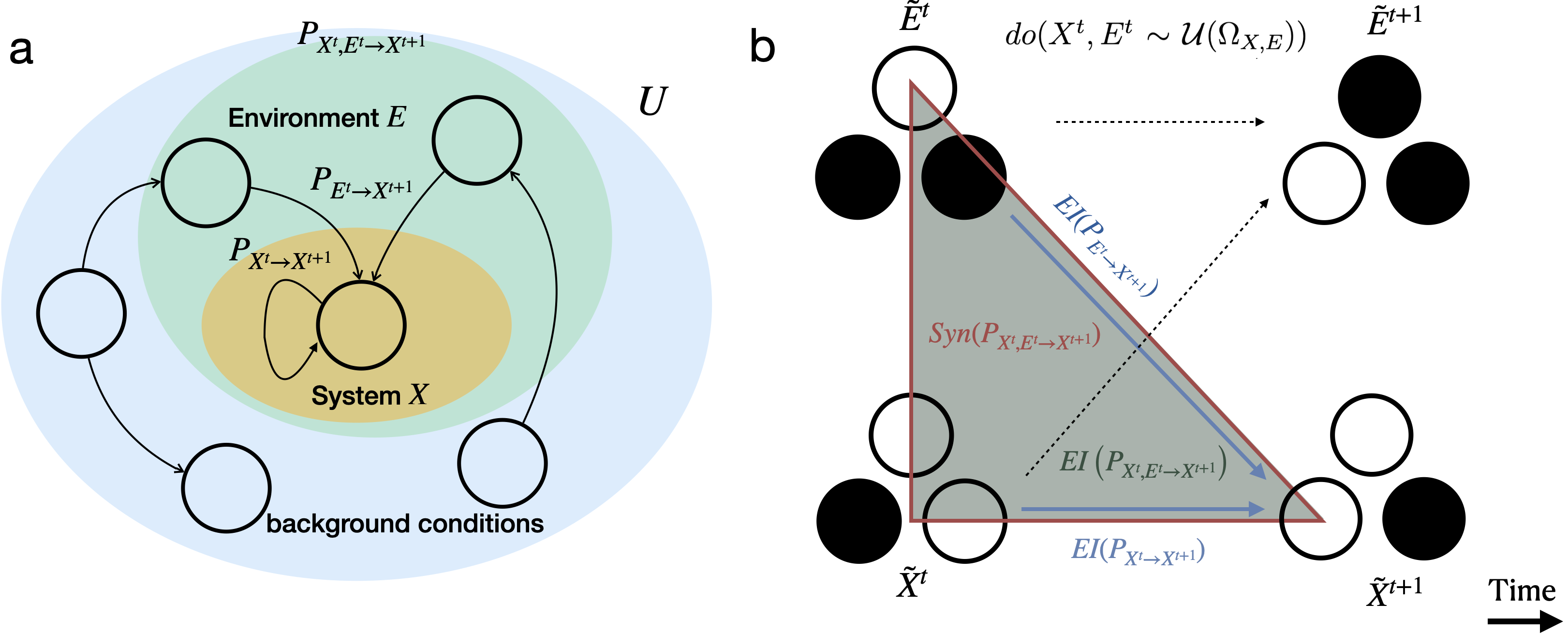}
    \caption{The causal diagram of the system's interaction with the environment. (a) When a set of variables is designated as the system, the spatial distinction between the system, environment, and background conditions is made, with arrows representing the causal relationships between variables. (b) The causal diagram of the interaction between $X$ and $E$ over time. Changes in the filling of circles in the diagram represent state transitions. In the diagram, the green triangle and the formula represent the EI from both the system and the environment to the system itself, the blue lines and the formulas represent the two types of unique information, and the red lines and the formula represent the effective synergistic information between the system and its environment, also known as \textit{flexibility}. This diagram is consistent with the causal diagram described in \cite{krakauer2020information}.}
    \label{fig:formu1}
\end{figure}

Referring to \cite{albantakis2023integrated}, we collectively term the system and environment as a \textit{Distinction}. We are solely concerned with the impact of the system and environment on the system's state at the next moment, thus we define the Distinct TPM as 
\begin{equation}
\label{eq:distinct_tpm}
    P_{X^t,E^t \to X^{t+1}}: P(X^{t+1}|X^t,E^t). 
\end{equation}
The state spaces of the joint variable \( X, E \) and of \( X \) are denoted as \( \Omega_{(X,E)} \) and \( \Omega_{X} \), respectively. Consequently, the shape of this TPM is \( |\Omega_{(X,E)}| \times |\Omega_{X}| \). In Figure \ref{fig:formu1}(b), we focus on the causal arrows from \(X^t\) and \(E^t\) pointing to \(X^{t+1}\), which corresponds to the dynamics described by the system's distinct TPM.

\subsection{Individual and External Driving Mechanisms}
Starting with the distinct TPM \( P_{X^t,E^t \to X^{t+1}} \), we can derive the TPM \( P_{X^t \to X^{t+1}} \), which considers only the system's internal dynamics and excludes environmental information, termed the \textit{Individual Mechanism}. Meanwhile, we extract the TPM \( P_{E^t \to X^{t+1}} \), which focuses exclusively on the environment's impact on the system's information, termed the \textit{External Driving Mechanism}. The definitions of these mechanisms are as follows, and their corresponding relationships with the causal diagrams are also marked in Figure \ref{fig:formu1}(a).
\begin{gather}
\label{eq:im}
     P_{X^t \to X^{t+1}}=|\Omega_E|^{-1}\sum_{E^t=e^t} p(X^{t+1}|X^t,E^t).\\
\label{eq:edm}
     P_{E^t \to X^{t+1}}=|\Omega_X|^{-1}\sum_{X^t=x^t} p(X^{t+1}|X^t,E^t).
\end{gather}
When we observe the probability distribution of environmental states, the conditional probability of the system can be obtained through $P(X^{t+1}\mid X^t)=\sum_{E^t=e^t} P(E^t)P(X^{t+1}\mid X^t,E^t)$. Because we introduce the do-operator \cite{pearl2018book}, which intervenes in both the environment and the system to achieve a uniform distribution, denoted as $do(X^t, E^t \sim \mathcal{U}(\Omega_{X,E}))$ in Figure \ref{fig:formu1}(b), we have $P(E^t)=|\Omega_E|^{-1}$. Consequently, we derive the expressions for $P_{X^t \to X^{t+1}}$ and, similarly, for $P_{E^t \to X^{t+1}}$ as Eqs. (\ref{eq:im}) and (\ref{eq:edm}) describe.

\subsection{Definition of Flexibility}
After obtaining those TPMs, we introduce the \(EI\) function for an arbitrary causal mechanism (TPM). \(EI\) represents the effective information, which quantifies the strength of causal effects for a TPM \cite{hoel2013quantifying,yuan2024emergence}. EI is defined as a mutual information for an intervened uniformly distributed input variable \( X \) and its corresponding output variable \( Y \) as shown in the following formula.

\begin{equation}
    \begin{aligned}
        EI(P_{X\to Y}) &\equiv  I(X,Y\mid do(X \sim \mathcal{U}(\Omega_{X})))\\
         &=\frac{1}{N}\sum^N_{i=1}\sum^N_{j=1}p_{ij}\log\frac{N\cdot p_{ij}}{\sum_{k=1}^N  p_{kj}},
    \end{aligned}
\end{equation}
where $p_{ij}$ is an element of the TPM $P_{X\to Y}$, and $N$ is the number of states of the input variable. For further details, see Appendix \ref{sec:ei}. We can define EIs for the mechanisms $P_{X^t\rightarrow X^{t+1}}$ and $P_{E^t\rightarrow X^{t+1}}$ which are also coined as Individual Driving Information and External Driving Information:
\begin{gather}
\label{eq:un1}
    EI(P_{X^t \to X^{t+1}})=I(X^t,X^{t+1}\mid do(X^t, E^t \sim \mathcal{U}(\Omega_{X,E}))),\\
\label{eq:un2}
    EI(P_{E^t \to X^{t+1}})=I(E^t,X^{t+1}\mid do(X^t, E^t \sim \mathcal{U}(\Omega_{X,E}))).
\end{gather}
These equations represent, respectively, the portion of information for the next moment of the system that is provided solely by the system itself and the portion provided solely by the environment, respectively. Naturally, due to the properties of mutual information, both types of effective information are non-negative. Within the effective joint mutual information, after the subtraction of these two components, what remains is the information that is exclusively provided by the system and the environment in union, termed as \textit{flexibility}, or the effective synergy between system and environment, denoted by \( Syn(P_{X^t,E^t\to X^{t+1}}) \). 
\begin{gather}
\label{eq:syn}
    Syn(P_{X^t,E^t\to X^{t+1}})=EI(P_{X^t,E^t\to X^{t+1}})-EI(P_{X^t \to X^{t+1}})-EI(P_{E^t \to X^{t+1}}).
\end{gather}
The correspondence between these indicators and the causal diagram is presented in Figure \ref{fig:formu1}(b). 

\subsection{Properties of Flexibility}
Actually, \( Syn(P_{X^t,E^t\to X^{t+1}}) \) aligns with the definitions and axiomatic system of PID theory \cite{williams2011information} regarding the requirements for synergistic information in trivariable systems. 

In the following paragraphs,  we denote \( \tilde{Z} \) as the intervened version of any random variable $Z$ after the intervention \( do(X^t, E^t \sim \mathcal{U}(\Omega_{X,E})) \). Consequently, we have the following theorem:
\begin{theorem}
\label{th:syn}
     In a trivariable system, the flexibility defined in Eq.(\ref{eq:syn}) is the synergistic information of $\tilde{X}^t, \tilde{E}^t$ with respect to $\tilde{X}^{t+1}$.
\end{theorem}
This theorem pertains to a specific PID axiomatic system(please refer to Appendix \ref{sec:pid}).
For the proof of this theorem, please refer to Appendix \ref{sec:proof1}. The upper bound of synergy is \(\min\{I(\tilde{X}^t; \tilde{X}^{t+1} | \tilde{E}^t), I(\tilde{E}^t; \tilde{X}^{t+1} | \tilde{X}^t)\}\). The proof of this property is provided in Appendix \ref{sec:proof2}. We can further decompose the synergy term, i.e., Equation \ref{eq:syn}. 
\begin{corollary}
\label{cor:decomp}
    The flexibility defined in Eq.(\ref{eq:syn}) can be decomposed into two components: Expansiveness and Introversion.
\end{corollary}
We define \textit{Expansiveness}, abbreviated as \textit{Exp}, and \textit{Introversion}, abbreviated as \textit{Int}, as follows:
\begin{gather}
\label{eq:exp}
    Exp(P_{X^t,E^t\to X^{t+1}})=\frac{1}{|\Omega_E|} \sum_{e \in \Omega_E} H\left( \frac{1}{|\Omega_X|} \sum_{x \in \Omega_X} P_{x,e} \right) +  \frac{1}{|\Omega_X|} \sum_{x \in \Omega_X} H\left( \frac{1}{|\Omega_E|} \sum_{e \in \Omega_E} P_{x,e} \right),\\
\label{eq:int}
    Int(P_{X^t,E^t\to X^{t+1}})=2\log_2{|\Omega_X|}- \frac{1}{|\Omega_{X,E}|} \sum_{x \in \Omega_X} \sum_{e \in \Omega_E} H(P_{x,e}) - H\left(\frac{1}{|\Omega_{X,E}|} \sum_{x \in \Omega_X} \sum_{e \in \Omega_E} P_{x,e} \right).
\end{gather}
The proof of Corollary \ref{cor:decomp} can be found in Appendix \ref{sec:proof3}.

To elucidate the meanings of expansiveness and introversion, we first introduce the EI function and its decomposition \cite{hoel2017map}.
\begin{equation}
\label{eq:ei_factor}
    EI(P_{X\to Y}) = \underbrace{ - \left( \frac{1}{N} \sum_{i=1}^N H(P_i) \right) }_{\text{determinism}} + \underbrace{ H\left( \frac{1}{N} \sum_{i=1}^N P_i \right) }_{\text{non-degeneracy}}
\end{equation}
Here, \( X \) denotes an arbitrary input variable with a state space of size \( N \), \( Y \) represents the corresponding output variable, \( P_i \) corresponds to the \( i \)-th row of the TPM \( P_{X \to Y} \), and \( H(P) \) indicates the Shannon entropy of probability distribution \( P \). The term $- \left( \frac{1}{N} \sum_{i=1}^N H(P_i) \right)$ represents determinism; when it is high, it indicates that the system has low noise. The term $H\left( \frac{1}{N} \sum_{i=1}^N P_i \right)$ represents non-degeneracy; when it is low(degeneracy is high), it implies that the system will deterministically converge to certain states, indicative of attractor dynamics.
\begin{figure}
    \centering
    \includegraphics[width=0.8\linewidth]{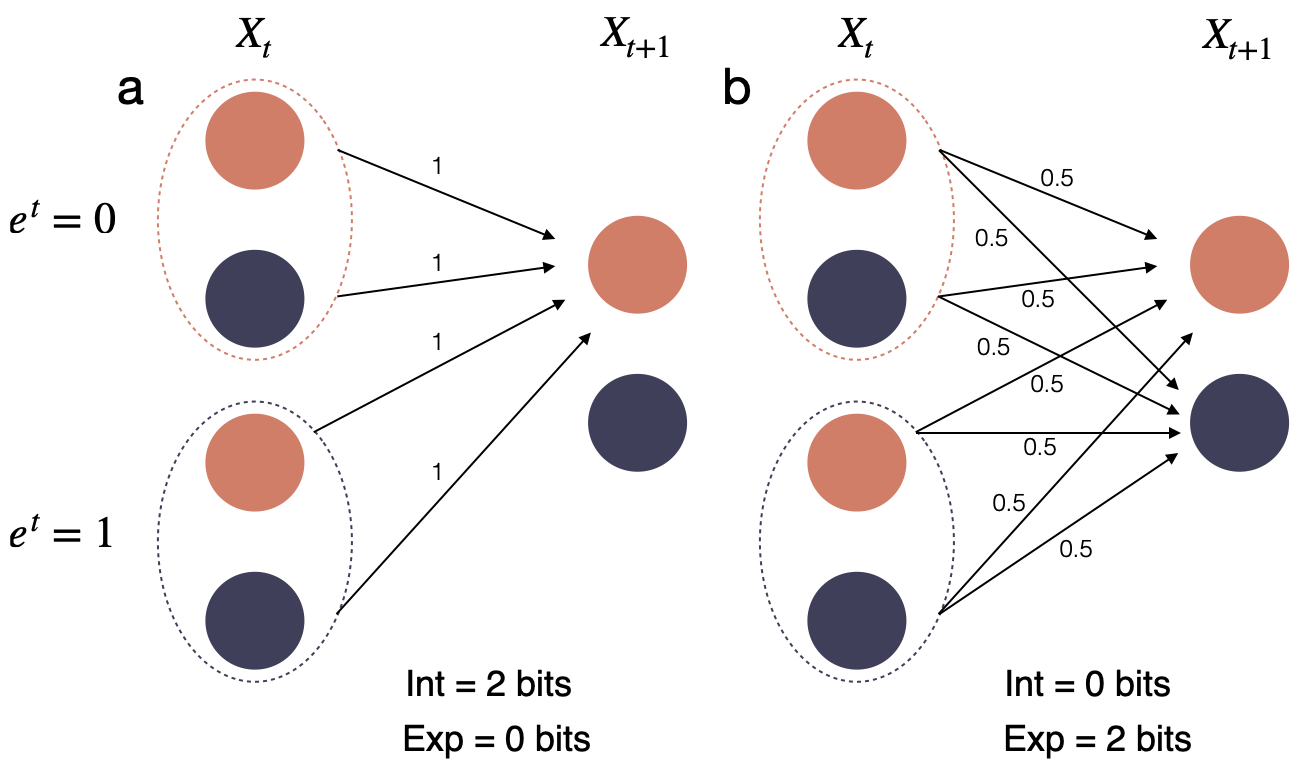}
    \caption{Schematic diagrams of expansiveness and introversion. Different colors of the circles represent different system states, while different colors of the dashed lines indicate different environmental states at that time. Arrows with associated numerical values denote transition probabilities from system-environment states at time \( t \) to system states at time \( t+1 \). (a) shows the case of low expansiveness and high introversion, while (b) shows the case of high expansiveness and low introversion.} 
    \label{fig:exp_int}
\end{figure}
When the system employs different TPMs corresponding to individual environmental states, the environmental context becomes explicitly incorporated. In contrast, when environmental states remain unspecified, the system adopts an environment-averaged TPM. Through the Effective Information (EI) decomposition framework, Equation \ref{eq:exp} measures the system's state-specific differentiation under well-defined environments versus its stochastic variability under environmental uncertainty. These dual aspects together constitute Expansiveness as the system's outward adaptive orientation. Correspondingly, Equation \ref{eq:int} captures the system's structured coordination in explicit environments versus its reduced differentiation in ambiguous contexts, jointly characterizing Introversion as the internal consolidation tendency. This demonstrates two distinct pathways through which the system enhances flexibility: expansiveness and introversion. The relationship between the magnitudes of Exp and Int under different conditions can be referred to in Figure \ref{fig:exp_int}.

Given that the Shannon entropy of the system's probability distribution ranges from 0 to \(\log_2{|\Omega_X|}\), we can determine the numerical ranges for Exp and Int: \(0 \leq Exp(P_{X^t,E^t\to X^{t+1}}), \allowbreak Int(P_{X^t,E^t\to X^{t+1}}) \leq 2\log_2{|\Omega_X|}\).


\section{\label{sec:3}Results}
In the subsequent experiments, we will validate the meaning and functionality of these metrics.
\subsection{Cellular Automaton}
Cellular automata (CA), particularly one-dimensional elementary CA, are used to simulate artificial life, with Wolfram identifying 256 possible rule sets \cite{wolfram1983statistical}. Wolfram classified CA into four behavior types: stable, periodic, chaotic, and complex, with Class IV potentially being computationally universal. However, this classification is subjective, leading Langton to introduce the parameter $\lambda = 1 - \frac{n}{8}$ \cite{langton1990computation} to quantify CA behaviors, where $n$ is the number of outputs being 1 in the CA rule table. He demonstrated that CA behaviors can continuously transition from Class I to Class III as $\lambda$ varies from $0$ to $1$, with values between 0.3 and 0.6 corresponding to complex Class IV dynamics. 
\begin{figure}
    \centering
    \includegraphics[width=1\linewidth]{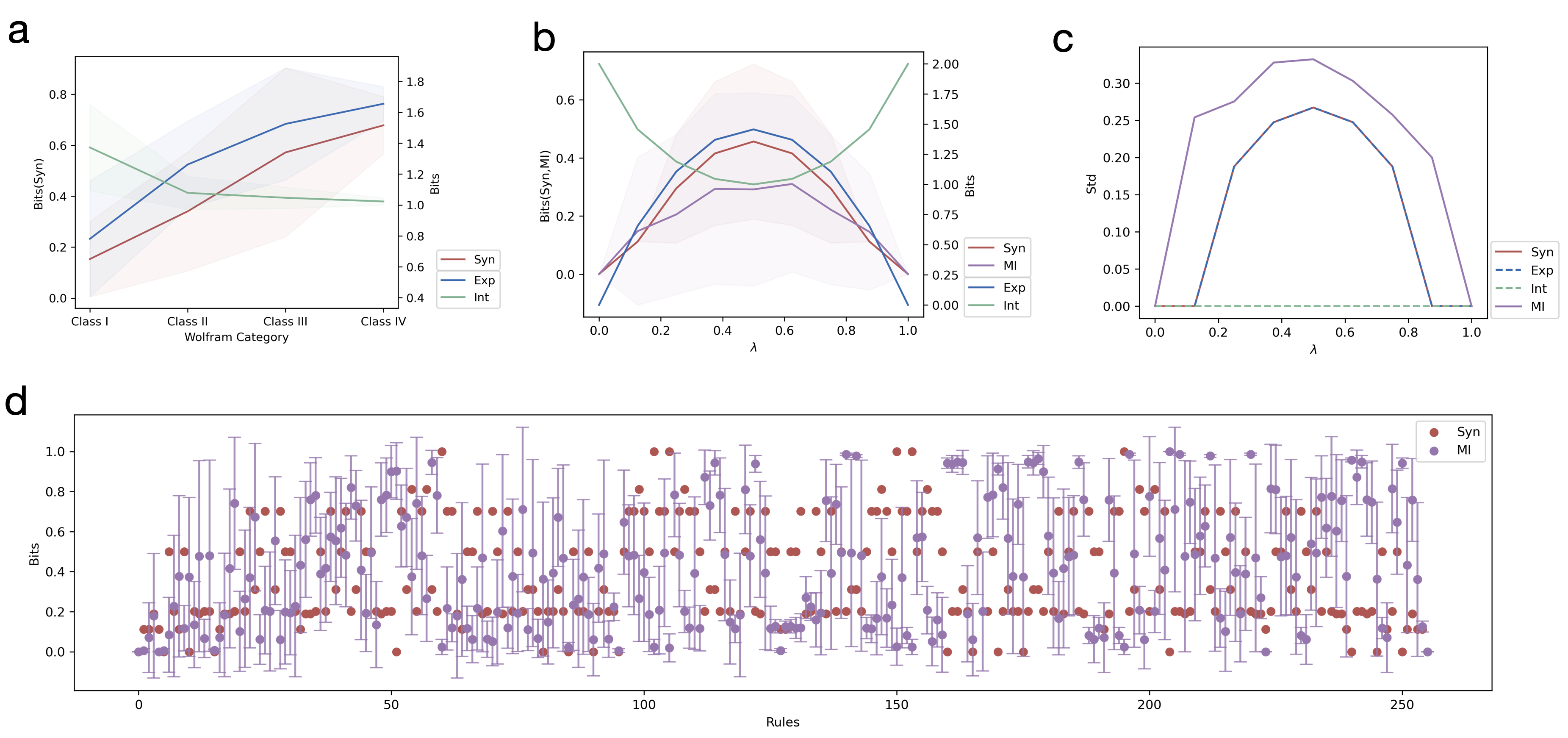}
    \caption{(a) The trend of expansiveness, introversion and flexibility between the system and the environment as the type of CA changes from Class I to Class IV. The line represents the mean, and the band represents the standard deviation. (b) A comparative trend chart of them, along with mutual information proposed by \cite{langton1990computation}, as $\lambda$ varies. (c) A comparative chart of the standard deviation for the four indicators calculated. (d) A scatter plot of flexibility and mutual information for all 256 CA. The flexibility of a specific rule-based cellular automaton is calculated without a standard deviation, whereas mutual information includes a standard deviation, depicted by error bars.
}
    \label{fig:ca}
\end{figure}
As depicted in Figure \ref{fig:ca}, we have validated the relationship between flexibility and its decomposition with the behavior types of CA. In Figure \ref{fig:ca}(a), as CA become increasingly complex, expansiveness rises while introversion declines, overall reflecting an increase in flexibility. A similar phenomenon is observed in Figure \ref{fig:ca}(b). When $\lambda$ is between 0.3 and 0.6, the increase in expansiveness exceeds the decrease in introversion, leading to an increase in flexibility. The following mathematical relationship indicates that, in noise-free CA, introversion is indeed a definite function of the $\lambda$.
\begin{equation}
\begin{aligned}
    Int(P_{X^t,E^t\to X^{t+1}})&=2\log_2{|\Omega_X|}-  \frac{1}{|\Omega_{X,E}|} \sum_{x \in \Omega_X} \sum_{e \in \Omega_E} H(P_{x,e}) - H\left(\frac{1}{|\Omega_{X,E}|} \sum_{x \in \Omega_X} \sum_{e \in \Omega_E} P_{x,e} \right)\\
    &=2\log_2{|\Omega_X|}- H\left(\frac{1}{|\Omega_{X,E}|} \sum_{x \in \Omega_X} \sum_{e \in \Omega_E} P_{x,e} \right)\\
    &=2\log_2{|\Omega_X|}-H(\lambda,1-\lambda)
\end{aligned}
\end{equation}

This explains why, in Figure \ref{fig:ca}(c), the standard deviation of introversion is zero, while the fluctuations in effective synergy are attributed to expansiveness. Concurrently, it is evident that mutual information, as a measure based on observational data, exhibits a larger standard deviation, even though its trend aligns with that of effective synergy. In Figure \ref{fig:ca}(d), we run iterations for 200 steps under each initial condition in a space of 10-cell automata. Based on these observational data, we calculate the mutual information from time t to t+1 for a single cell and compute the expectation and standard deviation over all possible initial conditions. The varying standard deviations across different rules of CA indicate that the selection of initial conditions is crucial for the computation of mutual information for some rules. In contrast, the calculation of effective synergy is independent of initial conditions.

\subsection{Identify Flexible Motifs in Gene Regulatory Networks}
We will measure the flexibility of gene regulatory networks (GRNs) using real data.

GRNs describe how a collection of genes governs key processes within a cell, which are often modeled as Boolean networks. Kadelka et al. \cite{kadelka2024meta} established the most comprehensive repository of expert-curated Boolean GRN models to date, encompassing both structural configurations and Boolean functions. These models describe the regulatory logic underlying a variety of processes in numerous species across multiple kingdoms of life. Due to computational constraints, our analysis was limited to a subset of these networks. We select 63 models, with node counts ranging from 5 to 67, encompassing animal, plant, fungal, and bacterial domains. 

To investigate which GRN structures exhibit enhanced environmental responsiveness in real-world settings, we assessed the flexibility of various three-node subgraph configurations in Figure \ref{fig:grn}(a). The analysis revealed that feedback loops (FBLs) demonstrated the highest flexibility values. In fact, FBLs in GRNs carry important biological functions. For example, FBL structures often exhibit dynamical compensation(DC), which is the ability of a model to compensate for variation in a parameter \cite{alon2019introduction}. Additionally, many oscillators in biological systems originate from negative FBLs, known as \textit{repressilators} \cite{alon2019introduction}. They play a crucial role in adapting to environmental changes and regulating their own cycles. In Figure \ref{fig:grn}(b), the mean flexibility of FBLs is also high, confirming the above conclusions. Moreover, the structure with the highest value among four-node structures is not a simple FBL but a more connected structure that includes FBLs (see the highlighted part in the figure). This structure has not yet been fully studied and named by biologists. Perhaps it carries some interesting functions related to biological adaptation to the environment that have yet to be discovered.

To verify that systems with higher flexibility have a stronger ability to respond to environmental changes, we compared the evolution of gene activation states under environmental shocks in Figure \ref{fig:grn}(c-d). The genes in (c) are from the GRN of macrophage activation, with a flexibility of 0.566, while those in (d) are from the GRN of tumour cell invasion and migration, with a flexibility of 0. Comparing the time series curves, the former exhibits a greater diversity of steady states under different environmental conditions. We measured more precisely the mean mutual information between consecutive moments when environmental state changes lead to shifts in system steady states, finding a positive correlation coefficient of 0.487 with flexibility. For further experimental details, see Appendix \ref{sec:grn}.
\begin{figure}
    \centering
    \includegraphics[width=1\linewidth]{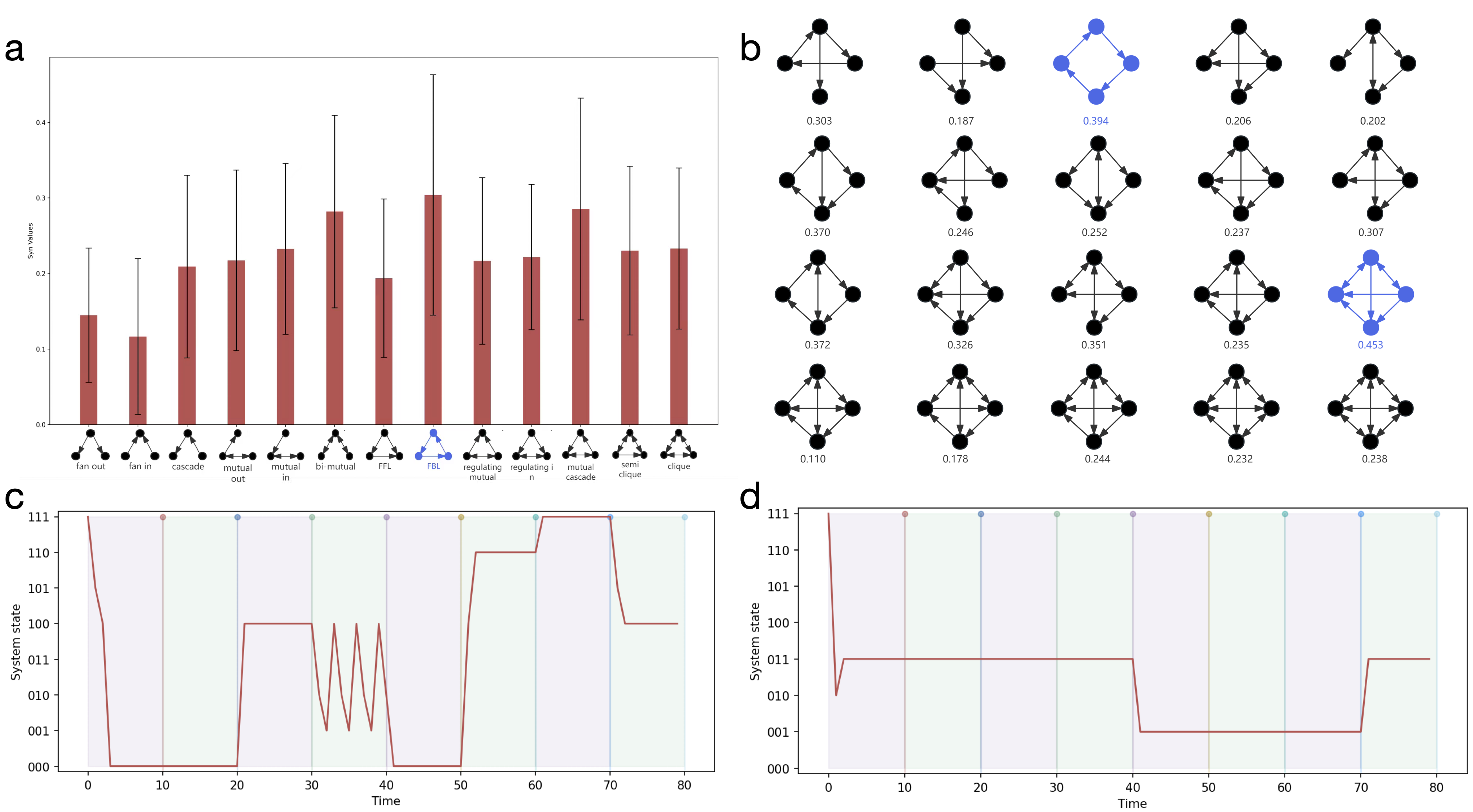}
    \caption{(a) The mean and standard deviation of flexibility for all possible three-node subgraph structures in various real environments, where the value for FBL is the highest. (b) The mean flexibility for some four-node subgraph structures, with the two highlighted structures having the highest values. (c) The time series of state changes for the system composed of the genes BAG4, BAG4\_TNFRSF1A, and TNF\_BAG4\_TNFRSF1A in the GRN of macrophage activation, starting from the initial state "111" and switching environmental states every 10 steps, with a total of 8 randomly selected environmental states. (d) The time series generated by the system composed of the genes CDH1, CDH2, and GF in the GRN of tumour cell invasion and migration, under experimental conditions consistent with those in (c).}
    \label{fig:grn}
\end{figure}

\subsection{Random Boolean Network and Machine Learning}
To explore dynamical characteristics influencing flexibility beyond network topology, subsequent experiments investigate the effects of dynamical parameters on random Boolean networks (RBNs) with fixed structures. Furthermore, while previous experiments were conducted under known dynamical mechanisms, practical scenarios frequently involve data-driven problems with unknown underlying mechanisms. We therefore employ RBN simulations integrated with machine learning methodologies. By first reconstructing governing mechanisms from observational data and subsequently performing flexibility measurements, we validate the applicability of our framework to systems with concealed dynamical rules.

\begin{figure}
    \centering
    \includegraphics[width=0.8\linewidth]{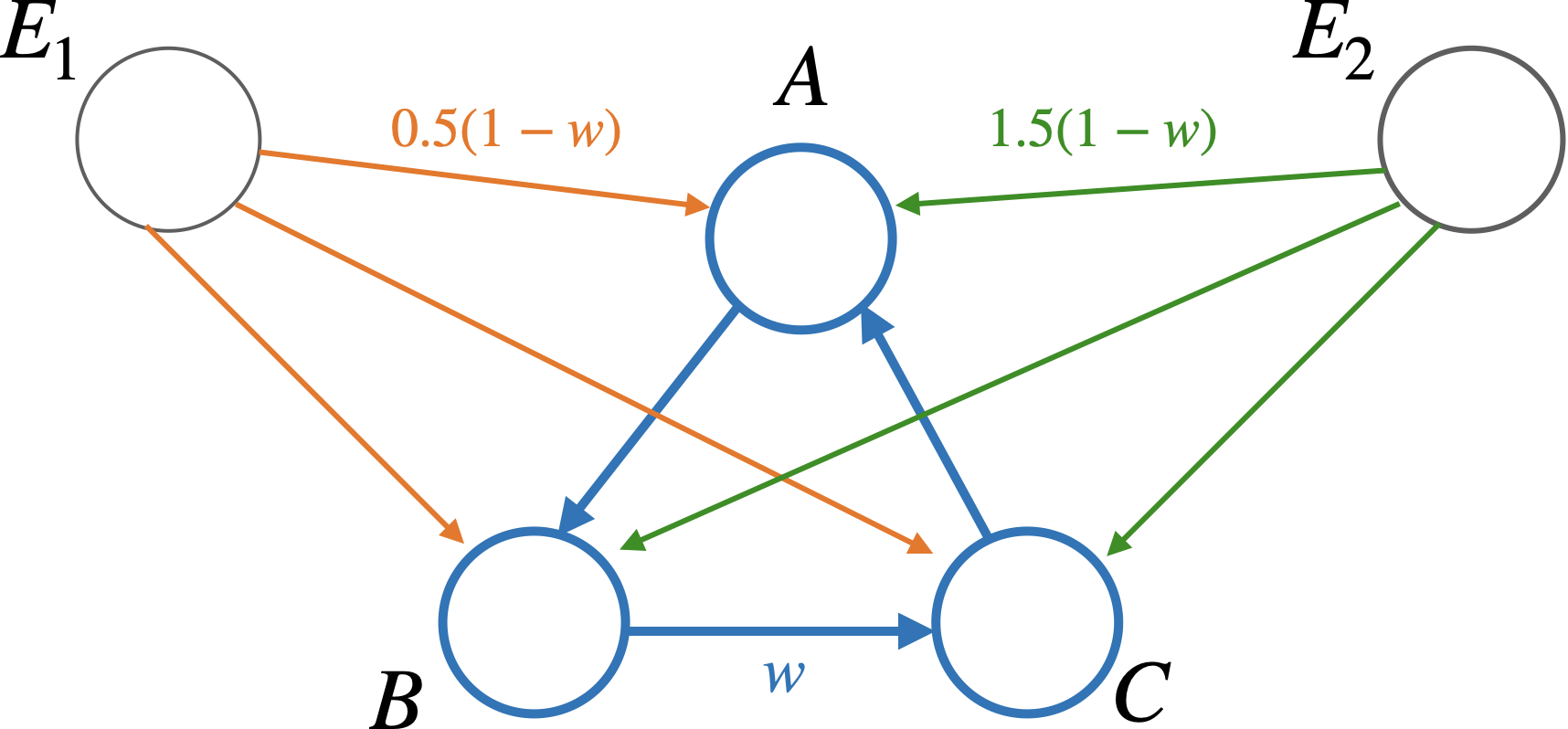}
    \caption{The schematic illustrates the experimental design framework, where nodes $A$, $B$, and $C$ constitute the core system interacting with environmental variables $E_1$ and $E_2$. Edges indicate interaction relationships, with weights $w_{j,i}$ encoded by color-coded mathematical symbols (colored circles denote self-loops). The continuously adjustable parameter $w \in [0,1]$ governs interaction intensities, whose functional role is defined through the mathematical formulation in Eq.\ref{eq:rbn_rule}.}
    \label{fig:rbn_exp}
\end{figure}

In Figure \ref{fig:rbn_exp}, each variable can take on two values, 0 or 1. For any variable $X_i$ in the system, its update rule is defined as follows:

\begin{equation}
\label{eq:rbn_rule}
    P(X_{i}^{t+1}=1\mid u^t)=\frac{1}{1+\exp(-k\sum_{j=1}^nw_{j,i}u_j^t)}
\end{equation}

Each variable $U_j$'s edge acting on another variable $X_i$ respectively corresponds to a weight value $w_{j,i} \in [0, 1]$. $k \in [0, +\infty]$ is a parameter controlling the noise magnitude. When $k = 0$, the noise is the greatest, that is, regardless of the values of the input variables, the conditional probability is a uniform distribution. The larger $k$ is, the smaller the noise intensity is. We set the temperature $T=\frac{1}{k}$, for $k\neq 0$. In this experiment, a total of two variables were set. One is the temperature $T$, and the other is the proportion of the system's own variables and the environmental variables' effect on the system. We set the weight of the system's own effect as $w \in [0, 1]$ (the solid line in Figure \ref{fig:rbn_exp}, including self-loops), and the weights of the effects of the environmental variables on the system (the dashed arrows in Figure \ref{fig:rbn_exp}) are $w_{E_1,i} = 0.5(1 - w)$ and $w_{E_2,i} = 1.5(1 - w)$.  It can be seen that the larger $w$ is, the stronger the influence of the system on itself is, and the weaker the influence of the environment on the system is. The trend graph of the flexibility varying with $T$ and $w$ is shown in Figure \ref{fig:rbn_result}(a). 

\begin{figure}
    \centering
    \includegraphics[width=1\linewidth]{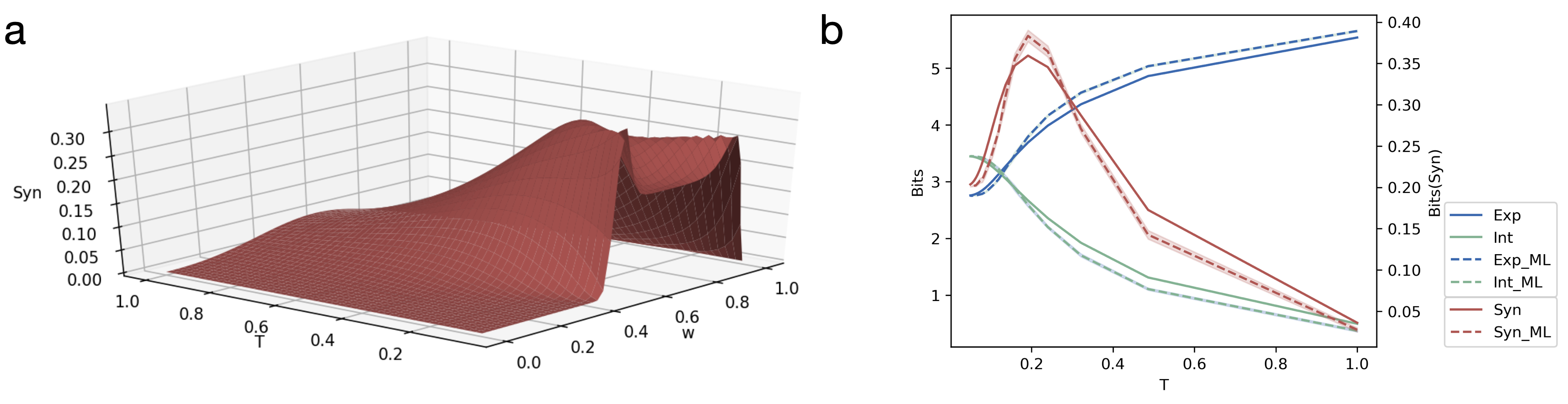}
    \caption{The experimental result graph. In (a), it indicates The trend graph of the flexibility varying with the temperature $T$ and the weight ratio parameter $w$. (b) shows that when $w = 0.5$, the changing trends of flexibility, expansiveness and introversion with respect to $T$. The solid line represents the result of the ground-truth, which is the curve intercepted by the section of $w = 0.5$ in (a). The dashed line is the calculation result of the trained artificial neural network based on machine learning from the generated data. The radius of the band is the standard deviation of the results of 10 repeated experiments.}
    \label{fig:rbn_result}
\end{figure}

As illustrated in Figure \ref{fig:rbn_result}(a), the system exhibits near-zero flexibility when either intrinsic self-influence or environmental influence dominates the dynamics. A maximum flexibility value emerges at optimal coupling strength \( w \), demonstrating the critical role of balanced interactions. Furthermore, increasing the temperature \( T \) generally induces monotonic reduction in flexibility across the parameter space. Notably, Figure \ref{fig:rbn_result}(b) reveals a counterintuitive phenomenon where moderate noise levels (\( T \approx 0.2 \)) paradoxically enhance flexibility to peak values. The phenomenon of enhanced synergistic information under low noise levels has been previously captured by other metrics \cite{orio2023dynamical}. Our experiments demonstrate that this phenomenon originates from interactions at the dynamical mechanism level. More importantly, through the decomposition of flexibility, we reveal that systems leverage low noise to achieve greater diversity and environmental sensitivity (\(Exp\)). The benefits of this trade-off outweigh the loss of intrinsic order (\(Int\)) caused by noise, resulting in an optimal noise level that maximizes flexibility. In biological systems, noise inherent in the interactions of genes regulating circadian rhythms maintains oscillatory behavior without decay \cite{alon2019introduction}. This suggests that other complex systems may benefit similarly from controlled noise levels, enabling optimal environmental responsiveness. 


To validate the framework's capability in reconstructing and quantifying system mechanisms under unknown dynamics, we train a neural network (NN) with four fully-connected layers (32→64→64→16→8) using data generated from conditional probabilities at $w=0.5$ under uniform input distribution. The NN architecture employs LeakyReLU activations in hidden layers and cross-entropy loss for one-hot encoded inputs. The alignment between predicted (dashed) and theoretical (solid) curves in Figure \ref{fig:rbn_result}(b) demonstrates successful extension of our measurement framework to data-driven scenarios through machine learning integration.

\section{\label{sec:4}Discussion}
Overall, for systems and environments that satisfy the Markov dynamics assumption, we defined on the TPM how to measure the synergistic influence of the system and the environment on the system - flexibility, which captures the flexibility of the system. It is neither completely determined by the dynamics of the system itself (different from Individual Driving Information) nor completely determined by the dynamics of the environment on the system (different from External Driving Information), but corresponds to the part in the dynamics of the system and the environment as a whole where the whole is greater than the sum of the parts. In the experiments of CA, we verified that complex cellular automata have higher flexibility. More importantly, on the Boolean network data of GRNs selected by experts, we found that the structure of feedback loops in various real environments has higher flexibility. This indicates that flexibility specifically points to the biological functions carried by this structure, such as dynamic compensation, biological cycle regulation, and so on \cite{alon2019introduction}. Currently, we only compared different Boolean network structures and have not distinguished different Boolean functions on the same structure. In the future, flexibility can be used to predict whether new structures and new dynamic functions have the biological functions we are interested in. 

Through decomposition of flexibility into dual components, we established expansiveness and introversion. These respectively measure the interaction variety between system and environment, and the level of dynamic organization in behavioral patterns. As shown in the machine learning experiments, the noise intensity in the dynamics is inversely proportional to the magnitude of introversion. While expansiveness measures, apart from the noise factor, to what extent the influences on the system from the system itself and the environment cannot be decoupled. Combining the machine learning experiments on RBNs and the computational results on CAs, we found that when there is noise in the dynamics, the reduction of noise increases flexibility by increasing the magnitude of introversion. And when the noise in the dynamics remains unchanged, the coupling degree of the influences of the system and the environment on the system is reflected by expansiveness, and at this time, the change of flexibility is dominated by the change of expansiveness. The analysis of expansiveness indicates that in the process of obtaining the variable space from the state space through a certain partition, the known boundary between the system and the environment is only one of several possible partitions, and synergy occurs when the original boundary is too ambiguous, so that we need to find a new coarse-graining of the state space. 

At present, there are still some areas for improvement in this framework. We assumed that the dynamics satisfy Markovianity, while many real-world problems need to be solved within a non-Markovian framework. Additionally, although we obtained indicators of multivariate information decomposition with good properties on the three-variable system with two variables acting on one variable, in the future, the calculation problem of synergistic information when the source variables reach three or more needs to be addressed. In this paper, the calculation and experiments of the indicators are based on discrete systems, so in the future, the definition and calculation method of flexibility on continuous systems also need to be further proposed. 

We currently assume that the dynamics are known. If the dynamics are unknown but the data is accessible, machine learning techniques can also be used to obtain the underlying dynamic mechanism first and then measure it. The machine learning in this paper is a preliminary attempt to prove its feasibility. In the future, we can introduce more complex neural network models to learn more complex dynamic mechanisms, and even take flexibility as the optimization goal to train artificial models with higher flexibility. 

\begin{acknowledgments}
We wish to acknowledge the support of Swarma Research and the assistance of Lifei Wang, a scientist from Swarma. 
\end{acknowledgments}

\appendix

\section{Partial Information Decomposition Theory}
\label{sec:pid}
Partial Information Decomposition (PID) is a theoretical framework designed to address the problem of multivariate information decomposition \cite{williams2010nonnegative}. To calculate the information transfer between multiple variables, people have proposed quantification indicators such as total correlation \cite{watanabe1960information} and interaction information \cite{te1980multiple}. However, they have very important flaws in application, such as not satisfying non-negativity \cite{yeung2008information}. For this reason, Williams et al. \cite{williams2010nonnegative} proposed the PID theory, decomposing joint mutual information into three types of information atoms: redundant information, unique information, and synergistic information. As shown in Figure \ref{fig:pid}, in a three-variable system, \textbf{redundant information} represents the part of information that either of the two source variables can provide to the target variable; \textbf{unique information} represents the part of information that only one source variable can provide and the other source variable cannot; \textbf{synergistic information} refers to the part of information that can only be provided when the two source variables are together.

\begin{figure}
    \centering
    \includegraphics[width=0.5\linewidth]{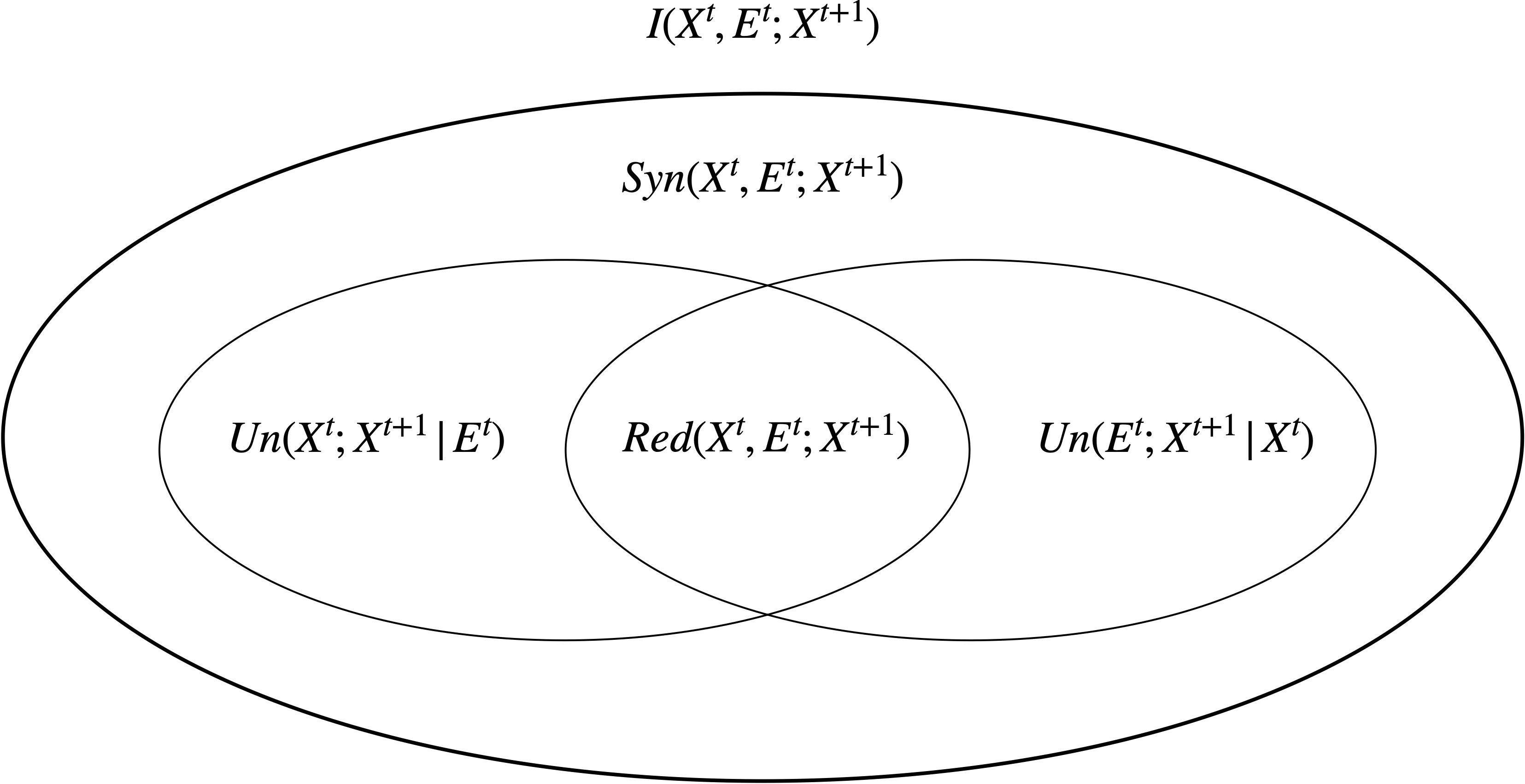}
    \caption{The PID framework in considering the information dynamics of the system and the environment. The two source variables are the variable $X^t$ of the system at time $t$ and the variable $E^t$ of the environment at time $t$, and the target variable is the variable $X^{t+1}$ of the system at time $t + 1$. The large outer ellipse represents the joint mutual information provided by the two source variables to the target variable. The overlapping part of the two small inner ellipses represents the redundant information provided by the two source variables. Each small ellipse represents the mutual information of a single source variable to the target variable. The remaining part other than the redundant information is the unique information. The area not covered by both small ellipses in the large ellipse is the synergistic information.}
    \label{fig:pid}
\end{figure}

The contribution of PID theory lies in that it not only gives a qualitative division but also provides a strict axiom system for redundant information \cite{williams2011information}, so that we can follow this axiom system to find quantitative calculation methods. The axiom system initially proposed by Williams et al. \cite{williams2011information} is as follows:
\begin{gather}
    (\mathbf{S})Red(A_1,A_2,...,A_k;S) \text{ is symmetric in } A_1,A_2,...,A_k.\\
    (\mathbf{I})Red(A;S)=I(A;S).\\
    (\mathbf{M})Red(A_1,A_2,...,A_k;S)\leq  Red(A_1,A_2,...,A_{k-1};S),\text{with equality if }A_{k-1}\subseteq A_k.
\end{gather}

Among them, $S$ is the target variable, $A_1,A_2,...,A_k\subseteq \{X_1, X_2,..., X_n\}$ is a combination of source variables, and $X_i$ is a source variable. Starting from these axioms, we can give the definitions of unique information and synergistic information, as well as the decomposition of mutual information and joint mutual information \cite{bertschinger2013shared}.

\begin{gather}
\label{eq:i_un_red}
    I(X_1;S) = Red(X_1, X_2;S)+Un(X_1;S|X_2)\\
    I(X_2;S) = Red(X_1, X_2;S)+Un(X_2;S|X_1)\\
\label{eq:i_syn}
    I(X_1,X_2; S) = Red(X_1, X_2;S)+Un(X_1;S|X_2)+Un(X_2;S|X_1)+Syn(X_1, X_2;S).
\end{gather}

In the above equation, $I(X_i; S)=Red(X_i;S)$, $ I(X_1,X_2; S)=Red(\{X_1,X_2\};S)$. Later, people found that these three axioms are not sufficient to describe concepts such as redundancy in our cognition, and the specific calculation formula for redundant information proposed by Williams et al. will yield abnormal results in some examples \cite{harder2013bivariate}. Therefore, on the basis of the above three axioms, people have added other axiomatic requirements \cite{bertschinger2013shared}.

\begin{gather}
    (\mathbf{LC})Red(A_1,A_2,...,A_k;SS') = Red(A_1,A_2,...,A_k;S)+Red(A_1,A_2,...,A_k;S'|S).\\
    (\mathbf{Id})Red(A_1,A_2;A_1\cup A_2)=I(A_1;A_2).
\end{gather}

In these equations, $Red(A_1,A_2,...,A_k;S'|S)=\sum_{s\in \mathcal{S}}p(s)Red(A_1,A_2,...,A_k;S'|s)$. $S'$ is also an arbitrary target variable. There are also some axioms mentioned in \cite{bertschinger2013shared} that will not be repeated here, because they can be derived from the above axioms. People have been trying to propose a computable definition of redundant information that can satisfy all the above axioms. For example, the calculation method proposed by Harder et al. \cite{harder2013bivariate} can satisfy axioms $(\mathbf{S})(\mathbf{I})(\mathbf{M})(\mathbf{Id})$, but it does not satisfy axiom $(\mathbf{LC})$. And the method proposed by Finn and Lizier \cite{finn2018pointwise} satisfies axiom $(\mathbf{LC})$, but does not satisfy axiom $(\mathbf{Id})$. Therefore, currently, there is no feasible calculation method for redundant, unique and synergistic information that can satisfy all axiom constraints.

\section{Effective Information}
\label{sec:ei}
Effective information (EI) is a quantitative measure that is used to characterize the strength of causal effects in markov dynamics. It was first proposed by Tononi et al.\cite{tononi2003measuring} as a key indicator in integrated information theory. Subsequently, Hoel et al. \cite{hoel2013quantifying} employed this metric to quantify the strength of causal effects in dynamics and further defined \textit{causal emergence}. EI is calculated based on the TPM and is independent of other factors. Its formal definition is as follows \cite{yuan2024emergence}:

\begin{equation}
\begin{aligned}
\label{eq:ei}
    EI(TPM) &= I(X^t,X^{t+1}|do(X^t\sim \mathcal{U}(\Omega_X)))  \\
    &= \frac{1}{N}\sum^N_{i=1}\sum^N_{j=1}p_{ij}\log\frac{N\cdot p_{ij}}{\sum_{k=1}^N  p_{kj}}
\end{aligned}
\end{equation}

In the given context, $\Omega_X$ represents the state space of $X^t$. The expression $do(X^t \sim \mathcal{U}(\Omega_X))$ indicates that the original definition of EI is the mutual information when the input variables are intervened to be uniformly distributed (i.e., at maximum entropy). It is essentially a function of the TPM \cite{zhang2024dynamical}. In the second equivalent computational formula, $p_{ij}$ denotes the probability of transitioning from state $i$ to state $j$, and $N = |\Omega_X|$. In the main text, Eqs.\ref{eq:un1}, \ref{eq:un2}, and \ref{eq:syn} all require the substitution of Eq. \ref{eq:ei} for calculation.

\section{Proofs}
\subsection{Proofs of Theorem \ref{th:syn}}
\label{sec:proof1}
In Appendix \ref{sec:pid}, we present the axiomatic system of PID theory. Based on this system, we initially state the following lemma:

\begin{lemma}
\label{lem:red0}
    Given axioms (\(\mathbf{S}\), \(\mathbf{I}\), \(\mathbf{M}\), \(\mathbf{LC}\), \(\mathbf{Id}\)), for any arbitrary variables \(X_1\), \(X_2\), and \(Y\), if \(X_1 \perp X_2\), the redundant information \(Red(X_1, X_2; Y) = 0\).
\end{lemma}
\begin{proof}
    From axiom \(\mathbf{LC}\), it follows that:
\begin{equation}
\begin{aligned}
\label{eq:thm_lc}
Red(X_1, X_2; (X_1, X_2, Y)) &= Red(X_1, X_2; (X_1, X_2)) + Red(X_1, X_2; Y | (X_1, X_2)) \\
&= Red(X_1, X_2; Y) + Red(X_1, X_2; (X_1, X_2) | Y).
\end{aligned}
\end{equation}

From the axioms (\(\mathbf{S}\), \(\mathbf{I}\), \(\mathbf{M}\)), it first follows that redundant information is non-negative and that redundant information is less than or equal to the joint mutual information provided by the two source variables \cite{williams2011information}. Consequently, conditional redundant information is also less than or equal to conditional joint mutual information. Thus, we can derive the following inequality:
\begin{equation}
    Red(X_1, X_2; Y | (X_1, X_2)) \leq I(X_1, X_2; Y | (X_1, X_2)) = 0.
\end{equation}

From the \(\mathbf{Id}\) axiom and the condition \(X_1 \perp X_2\), we can deduce that:
\begin{equation}
    Red(X_1, X_2; (X_1, X_2)) = I(X_1, X_2) = 0.
\end{equation}

Thus,
\begin{equation}
\label{eq:thm_zero}
    Red(X_1, X_2; Y) + Red(X_1, X_2; (X_1, X_2) | Y) = 0.
\end{equation}

Given that redundant information is non-negative, we have:
\begin{equation}
    Red(X_1, X_2; Y), Red(X_1, X_2; (X_1, X_2) | Y) \geq 0.
\end{equation}

Combining this with Eq.(\ref{eq:thm_lc}) and Eq.(\ref{eq:thm_zero}), we conclude that:
\begin{equation}
    Red(X_1, X_2; Y) = Red(X_1, X_2; (X_1, X_2) | Y) = 0.
\end{equation}

\end{proof}

Next, we restate the content of Theorem \ref{th:syn} and provide a proof:

In a trivariable system, the flexibility defined in Eq. \eqref{eq:syn} is the synergistic information of $\tilde{X}^t$ and $\tilde{E}^t$ with respect to $\tilde{X}^{t+1}$.

\begin{proof}
    Given that we intervene $X^t$ and $E^t$ to achieve a uniform distribution, resulting in $\tilde{X}^t$ and $\tilde{E}^t$, it follows that $\tilde{X}^t\perp \tilde{E}^t$. Based on Lemma \ref{lem:red0}, we have:
\begin{equation}
\label{eq:red_fin}
    Red(\tilde{X}^t, \tilde{E}^t; \tilde{X}^{t+1}) = 0.
\end{equation}

Drawing on the relationship between information atoms and mutual information, as given in Eq. \eqref{eq:i_un_red}:
\begin{equation}
    Un(\tilde{X}^t; \tilde{X}^{t+1} | \tilde{E}^t) = I(\tilde{X}^t; \tilde{X}^{t+1}) - Red(\tilde{X}^t, \tilde{E}^t; \tilde{X}^{t+1}) = I(\tilde{X}^t; \tilde{X}^{t+1}).
\end{equation}

The same applies to $Un(\tilde{E}^t; \tilde{X}^{t+1} | \tilde{X}^t)$. Since $\tilde{E}^t$ is also uniformly distributed, according to the definition and expression of effective information (please refer to Appendix \ref{sec:ei}), we obtain:
\begin{equation}
    I(\tilde{X}^t; \tilde{X}^{t+1}) = EI(P_{X^t \to X^{t+1}}).
\end{equation}

Based on Eqs. (\ref{eq:i_syn}) and (\ref{eq:red_fin}), we derive the following:
\begin{equation}
    \begin{aligned}
        Syn(\tilde X^t, \tilde E^t; \tilde X^{t+1}) &= I(\tilde{X}^t, \tilde E^t; \tilde{X}^{t+1}) - I(\tilde{X}^t; \tilde{X}^{t+1}) - I(\tilde{E}^t; \tilde{X}^{t+1}) \\
        &= EI(P_{X^t,E^t \to X^{t+1}}) - EI(P_{X^t \to X^{t+1}}) - EI(P_{E^t \to X^{t+1}}).
    \end{aligned}
\end{equation}

This is precisely the flexibility defined in Eq. (\ref{eq:syn}) in the main text.
\end{proof}

\subsection{Proofs of the Upper Bound of Synergy}
\label{sec:proof2}
We restate the property as follows: The upper bound of \( Syn(\tilde{X}^t, \tilde{E}^t; \tilde{X}^{t+1}) \) is \(\min\{I(\tilde{X}^t; \tilde{X}^{t+1} | \tilde{E}^t), \allowbreak I(\tilde{E}^t; \tilde{X}^{t+1} | \tilde{X}^t)\}\).
\begin{proof}
    According to the definition of EI,
    \begin{equation}
        EI(P_{X^t \to X^{t+1}}) = I(\tilde{X}^t; \tilde{X}^{t+1}),
    \end{equation}
    while still satisfying \( do(X^t, E^t \sim \mathcal{U}(\Omega_{X,E})) \). Therefore, by the chain rule of mutual information, we have
    \begin{equation}
        \begin{aligned}
        \label{eq:syn_eq}
            Syn(\tilde{X}^t, \tilde{E}^t; \tilde{X}^{t+1}) &= I(\tilde{X}^t, \tilde{E}^t; \tilde{X}^{t+1}) - I(\tilde{X}^t; \tilde{X}^{t+1}) - I(\tilde{E}^t; \tilde{X}^{t+1}) \\
            &= I(\tilde{X}^t; \tilde{X}^{t+1} | \tilde{E}^t) - I(\tilde{X}^t; \tilde{X}^{t+1})
        \end{aligned}
    \end{equation}
    Due to the non-negativity of mutual information, it follows that \( Syn(\tilde{X}^t, \tilde{E}^t; \tilde{X}^{t+1}) \leq I(\tilde{X}^t; \tilde{X}^{t+1} | \tilde{E}^t) \). Similarly, \( Syn(\tilde{X}^t, \tilde{E}^t; \tilde{X}^{t+1}) \leq I(\tilde{E}^t; \tilde{X}^{t+1} | \tilde{X}^t) \). 
    
    Thus, \(\min\{I(\tilde{X}^t; \tilde{X}^{t+1} | \tilde{E}^t), I(\tilde{E}^t; \tilde{X}^{t+1} | \tilde{X}^t)\}\) is the upper bound of \( Syn(\tilde{X}^t, \tilde{E}^t; \tilde{X}^{t+1}) \).
\end{proof}

\subsection{Proofs of Corollary \ref{cor:decomp}}
\label{sec:proof3}
We first restate the definition of flexibility, namely Eq. \ref{eq:syn},
\begin{gather}
    Syn(P_{X^t,E^t\to X^{t+1}})=EI(P_{X^t,E^t\to X^{t+1}})-EI(P_{X^t \to X^{t+1}})-EI(P_{E^t \to X^{t+1}}).
\end{gather}
as well as the expression for EI, namely Eq. \ref{eq:ei_factor},
\begin{equation}
    EI(P_{X\to Y}) = - \left( \frac{1}{N} \sum_{i=1}^N H(P_i) \right)  + H\left( \frac{1}{N} \sum_{i=1}^N P_i \right) 
\end{equation}

Substituting Eq. (\ref{eq:ei_factor}) into Eq. (\ref{eq:syn}), we have
\begin{equation}
    \begin{aligned}
        Syn(P_{X^t,E^t\to X^{t+1}}) &=  -\frac{1}{|\Omega_{X,E}|} \sum_{x \in \Omega_X} \sum_{e \in \Omega_E} H(P_{x,e}) + H\left(\frac{1}{|\Omega_{X,E}|} \sum_{x \in \Omega_X} \sum_{e \in \Omega_E} P_{x,e} \right) \\
        &- \left(-\frac{1}{|\Omega_X|} \sum_{x \in \Omega_X} H\left( \frac{1}{|\Omega_E|} \sum_{e \in \Omega_E} P_{x,e} \right) + H\left(\frac{1}{|\Omega_{X,E}|} \sum_{x \in \Omega_X} \sum_{e \in \Omega_E} P_{x,e} \right)\right) \\& - \left(-\frac{1}{|\Omega_E|} \sum_{e \in \Omega_E} H\left( \frac{1}{|\Omega_X|} \sum_{x \in \Omega_X} P_{x,e} \right) + H\left(\frac{1}{|\Omega_{X,E}|} \sum_{x \in \Omega_X} \sum_{e \in \Omega_E} P_{x,e} \right)\right)\\
        &= \underbrace{\frac{1}{|\Omega_E|} \sum_{e \in \Omega_E} H\left( \frac{1}{|\Omega_X|} \sum_{x \in \Omega_X} P_{x,e} \right) +  \frac{1}{|\Omega_X|} \sum_{x \in \Omega_X} H\left( \frac{1}{|\Omega_E|} \sum_{e \in \Omega_E} P_{x,e} \right)}_{\text{expansiveness}}\\
        &\underbrace{2\log_2{|\Omega_X|}- \frac{1}{|\Omega_{X,E}|} \sum_{x \in \Omega_X} \sum_{e \in \Omega_E} H(P_{x,e}) - H\left(\frac{1}{|\Omega_{X,E}|} \sum_{x \in \Omega_X} \sum_{e \in \Omega_E} P_{x,e} \right)}_{\text{introversion}} - 2\log_2{|\Omega_X|}\\
        &=Exp(P_{X^t,E^t\to X^{t+1}}) + Int(P_{X^t,E^t\to X^{t+1}}) - 2\log_2{|\Omega_X|}.
    \end{aligned}
\end{equation}

\section{Additional Experiments with GRN Data}
\label{sec:grn}
To illustrate the intuitive connection between flexibility and a system's flexible response to environmental changes, we plotted the scatter diagram shown in Figure \ref{fig:grn_app}. Here, MIs represent the average magnitude of mutual information between consecutive moments for various experimental conditions. It takes into account the impact of environmental changes on the system, accumulating mutual information only when environmental alterations lead to shifts in system steady states; otherwise, the system's mutual information under the new environment is recorded as 0. Additionally, since it measures the system's own mutual information, it quantifies whether the system maintains maximal intrinsic information transfer across any environment. Although the calculation of MIs is affected by sampling and does not have an exact correspondence with flexibility, the trend of the fitted line indicates a positive correlation between the two. Their Pearson correlation coefficient is 0.487. Conducting a hypothesis test with the null hypothesis of no significant correlation between the two yields a p-value less than 0.05, indicating a significant correlation. The network structure and function settings are derived from the GRN controlling apoptosis as described in the data from \cite{kadelka2024meta}.

\begin{figure}
    \centering
    \includegraphics[width=0.5\linewidth]{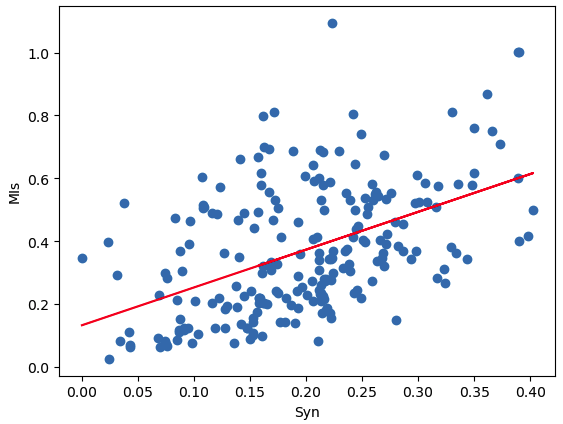}
    \caption{We searched all three-node subgraph structures in the GRN of apoptosis and calculated their flexibility (Syn) as well as the mean mutual information (MIs) between consecutive moments when environmental state changes lead to shifts in system steady states. The calculation of MIs considered all possible initial states of the system and all possible environmental states, with the sequence of environmental state transitions being randomly determined each time. This allowed us to generate a scatter plot of MIs versus Syn, with the red line representing the fitted linear regression of the plot.}
    \label{fig:grn_app}
\end{figure}


\bibliography{apssamp}
\end{document}